\documentclass{elsarticle}
\usepackage[utf8]{inputenc}
\usepackage{algorithmic}
\usepackage{algorithm2e}
\usepackage{graphicx}

\usepackage{subcaption}
\usepackage{amsmath}
\usepackage{amsthm}
\usepackage{url}
\usepackage{slantsc} 
\usepackage{amssymb}
\newtheorem{problem}{Problem}
\newtheorem{theorem}{Theorem}
\newtheorem{corollary}{Corollary}

\begin{document}

\title{A decision support system for optimizing the cost of social distancing in order to stop the spread of COVID-19}
\author[1,2]{Alexandru Popa}
\ead{alexandru.popa@fmi.unibuc.ro}
\address[1]{Department of Computer Science, University of Bucharest}
\address[2]{National Institute of Research and Development in Informatics}

\begin{abstract}
Currently there are many attempts around the world to use computers, smartphones, tablets and other electronic devices in order to stop the spread of COVID-19. Most of these attempts focus on collecting information about infected people, in order to help healthy people avoid contact with them. However, social distancing  decisions are still taken by the governments empirically. That is, the authorities do not have an automated tool to recommend which decisions to make in order to maximize social distancing and to minimize the impact for the economy. 

In this paper we address the aforementioned problem and we design an algorithm that provides social distancing methods (i.e., what schools, shops, factories, etc. to close) that are efficient (i.e., that help reduce the spread of the virus) and have low impact on the economy.

On short: a) we propose several models (i.e., combinatorial optimization problems); b) we show some theoretical results regarding the computational complexity of the formulated problems; c) we give an algorithm for the most complex of the previously formulated problems; d) we implement and test our algorithm; and e) we show an integer linear program formulation for our problem.
\end{abstract}


\maketitle

\section{Introduction}
\label{sec:intro}
The rapid spread of COVID-19 around the world is stunning. This novel coronavirus created an uprecedented lockdown in many countries which, in turn, caused an immense economic and social impact. Thus, many researchers investigate methods to stop this epidemic as soon as possible. For example, the list of papers on COVID-19 collected by the World Health Organization~\cite{WHO} contains around $10\,000$ publications, a huge number, given that the virus was first discovered in January 2020. The struggle involves researchers from various fields such as bioinformatics, epidemiology, sociology, mathematics and computer science.

One key factor to stop the spread of the virus is the \emph{social distancing} (see, e.g.,~\cite{glass2006targeted}).  Many companies and  organizations try to develop applications to aid the social distancing (see~\cite{TheGov} for a long list of current such projects). However,  many applications seem to focus on tracking people movement. To the best of our knowledge, we do not know any aplication that advises the authorities which decisions to make. As Thomas Pueyo wrotes in his article published on the 19th of March 2020~\cite{Pueyo2020} (Chart 16), governments should have a chart with the effect and the cost of various social distancing measures. As it is currently observed in the world (and especially in Europe), many governments were afraid to take severe social distancing measures in order to avoid a high economic loss.

In this paper we try to address this issue as follows. We first build a model (i.e., a combinatorial optimization problem) that captures the current setting: the risk of the people to get COVID-19, the contact between various people and the cost of closing various facilities such as schools, parks, cities, factories, etc.. We show that the problem we introduce is NP-hard (as it is often the case with complex combinatorial optimization problems). Then, since we cannot solve the problem exactly in polynomial time, we provide a heuristic polynomial time algorithm for this problem. To understand the performance of our algorithm, we implement and test it in Section~\ref{sec:experiments}. We generate our test data  using special probability distributions that simulate real world social networks as we present in Section~\ref{sec:data_gen}. Our experiments are encouraging and show that even with a $1\%$ budget (from the total cost of locking down the entire country), we can reduce the population risk by more than $5$ times compared with the situation in which no measures are taken. Thus, we show that there is a possibility for a ``beautiful'' lockdown that is efficient in the fight with COVID-19 and safe for the economy. 

We mention that the paper is written in a bottom-up fashion. More precisely, in Section~\ref{sec:frame1} we present a prelimiary model that we designed in the early stages of our study.  Even if we do not consider the model in Section~\ref{sec:frame1} further in the paper, the motivation for introducing it is two-fold. Firstly, by presenting the model in Section~\ref{sec:frame1} we show the reader the complete path we took to design the model in Section~\ref{sec:frame2} (instead of simply presenting the final product). Secondly, researchers who aim to study and improve the models presented in this paper may find useful to understand the difficulty behind designing a comprehensive model.

The paper is structured as follows. At first, in Section~\ref{sec:frame1} we present the first set of problems that aim to model the problem. We also show that these problems are NP-hard. Then, in Section~\ref{sec:frame2} we present our actual framework. In Section~\ref{sec:algo} we design an algorithm for the problem presented in Section~\ref{sec:frame2}. Then, in Section~\ref{sec:experiments} we describe our experiments. In Section~\ref{sec:ilp} we present an integer linear program formulation for the problem introduced in Section~\ref{sec:frame2}. Finally, in Section~\ref{sec:conclusions} we discuss several directions for future work.

\section{Preliminary ideas}
\label{sec:frame1}

In this section we introduce a preliminary model (i.e., a collection of related combinatorial optimization problems) that helped us to derive the model from Section~\ref{sec:frame2}. 

\subsection{A tentative framework}
\label{sec:subframe1}
The input consists of an undirected \emph{complete} graph $G = (V,{V \choose 2})$ and a function $p : {V \choose 2} \to [0,1]$. 
Each node $v \in V$ in a graph corresponds to a person and $p(u,v)$ is the probability that two people get in contact to each other. Moreover, each vertex $v \in V$ has associated two values, $risk:V \to [0,1]$ and $vulnerability:V \to [0,1]$, representing how likely is a person to spread the disease (e.g., it can be $1$ if a person is tested positively with COVID-19 or close to 1 if a person was recently in a ``red area''), respectively how vulnerable is a certain person (e.g., there are studies showing that eldery people and people with chronic diseases are more likely to be affected).

Besides the input graph we are given: $k_1$ sets of vertices $V^1 = \{V^1_1, V^1_2, \dots, V^1_k\}$ each one having associated a value  $c_1 : \{1, 2, \dots, k_1\} \to \mathbb{R}_+$ and a value  $r_1: \{1, 2, \dots, k_1\} \to [0,1]$. The cost $c_1(i)$ represents the cost of reducing the value of all $p(a,b)$, $\forall a,b\in V^1_i$ with to $p(a,b) \cdot r_1(i)$ . Informally, the cost $c_1(i)$ represents the cost of closing facility $i$ (i.e., a school, a bar, restaurant, theater, etc.), which, in turn, reduces the probability of the people that belong to that facility to interact with each other. In a simple variant, each $r_1(i)$ can be set to 0, representing that two people that belong to that facility will have probability $0$ to interact once the facility is closed. 

Then, we have $k_2$ sets of vertices $V^2 = \{V^2_1, V^2_2, \dots, V^2_k\}$ each one with a value $c_2 : \{1, 2, \dots, k\} \to \mathbb{R}_+$ and a value $r_2: \{1, 2, \dots, k_2\} \to [0,1]$. The cost $c_2(i)$ represents the cost of reducing the value of all $p(a,b)$, where $a\in V^2_i$ and $b \notin V^2_i$ to $p(a,b)\cdot r_2(i)$. Informally, the cost $c_2(i)$ represents the cost of isolating the people in the group $V^2_i$ (for example, quarantining persons, small groups or even closing entire cities).

\subsection{Possible combinatorial optimization problems}
\label{sec:funny_problems}
Now we introduce a couple of objective functions and constraints that aim to model the current scenario. The overall goal is to reduce the spread of the virus while keeping the cost at a minimum. The first group of problems consider a simplified variant of the framework, ignoring the vulnerability and the risk of each person.

In the first problem the goal is to optimize the economic cost of closing various facilities and isolating various groups of people, while maximizing the number of components created. 
\begin{problem}
\label{prob1}
We are given a budget $B \in \mathbb{R}_+$ and a threshold $P \in [0,1]$. The goal is to select a set $\hat{V^1} \subseteq V^1$ and a set $\hat{V^2} \subseteq V^2$ such that the following two conditions are met:

\begin{enumerate}
    \item 
$$ \sum_{i \in \hat{V^1}} c_1(i) + \sum_{i \in \hat{V^2}} c_2(i) \le B$$
\item 
After the sets of facilities $\hat{V^1}$ and $\hat{V^2}$ are selected and the corresponding edges have their probabilities decreased (as described in Subsection~\ref{sec:subframe1}), we remove all the edges $(a,b) \in {V \choose 2}$ such that $p(a,b) \le P$. The goal is to maximize the number of connected components in the remaining graph.
\end{enumerate}

\end{problem}

As we stated above, the model does not consider all the information. However, it is useful in cases where not much data is available to conduct preliminary tests. Moreover Problem~\ref{prob1} is interesting to study from the theoretical point of view since it is a novel combinatorial optimization problem.

Notice that even this oversimplified variant of the framework is NP-hard since it is a generalization of the classical Vertex Cover problem as we show in Subsection~\ref{sec:hard}.

The second problem that we introduce is similar to the first problem, where the goal is to minimze the budget, while requiring for at least a certain number of connected components to be created.

\begin{problem}
\label{prob2}
We are given a number of desired connected components $N$ and a threshold $P \in [0,1]$. The goal is to select a set $\hat{V^1} \subseteq V^1$ and a set $\hat{V^2} \subseteq V^2$ such that the following holds. After the sets of facilities $\hat{V^1}$ and $\hat{V^2}$ are selected and the corresponding edges have their probabilities decreased, we remove all the edges $(a,b) \in {V \choose 2}$ such that $p(a,b) \le P$. The number of connected components in the remaining graph should be at least $N$. The goal is to minimize $$ \sum_{i \in \hat{V^1}} c_1(i) + \sum_{i \in \hat{V^2}} c_2(i)$$
\end{problem}

If we ask to maximize only the number of connected components we might obtain a solution that does not match the original motivation. For example, we can obtain a solution where we have many small components and a huge component, which is, of course, not desired in practice. Thus, we introduce the following two problems, in which we impose a restriction on the size of the connected componets resulted after the closure of facilities.

\begin{problem}
\label{prob3}
The input is the same as in Problem~\ref{prob1}. The goal is to minimize the number of nodes of the largest connected components in the remaining graph.
\end{problem}

\begin{problem}
\label{prob4}
The input is the same as in Problem~\ref{prob2}, but $N$ instead of being the number of connected componets desired is the maximum allowed size of a connected component. Thus, the goal is to choose a set of facilities of minimum total budget (if such a set exists) such that, after closing these facilities, each resulting component has size less than or equal to $N$.
\end{problem}

In the end of this section, we formulate two more complex problems that aim to take into considerations all the restrictions, including the \emph{risk} and the \emph{vulnerability}.

\begin{problem}
\label{prob5}
Besides the input graph and the data associated with the facilities, we are given a budget $B$, a threshold $P$ and two real numbers $W$ and $R$. We have the following constraints associated with the connected components resulted after closing the facilities:

\begin{enumerate}
    \item For any connected component $X$ the $\sum_{v \in X} vulnerability(v) \le W$. Informally, this constraint aims to avoid large groups formed by vulnerable people (such as eldery, or immunosuppressed).
    \item For any connected component $X$, the $\sum_{u,v \in X} (\max\{risk(u),risk(v)\} \cdot 1 - \min\{risk(u),risk(v)\}) \le R$. Informally, this constraint aims to avoid a connected component that mixes ``healthy'' and ``ill'' people. Notice that if two people have high risk (i.e., that are very likely to have COVID-19) or if two people have very low risk, then $\max\{risk(u),risk(v)\} \cdot 1 - \min\{risk(u),risk(v)\}$ is very close to $0$.
\end{enumerate}

The goal is to select a set of facilities such that, after removing the edges with the probability less than $P$, minimises the number of connected components that violate any of the two above mentioned constraints.
\end{problem}

The final problem that we propose in this section, is very similar to Problem~\ref{prob5} but aims to enforce that all the components resulted obey the restrictions. Nevertheless, in this variant, we are not given a constraint on the budget. Otherwise, if we are given  a constraint on the budget and on the connected components, it is NP-hard even to decide if a feasible solution exists (we obtain an instance of the Knapsack problem that is NP-hard~\cite{garey1979computers}).

\begin{problem}
\label{prob6}
The input is similar to Problem~\ref{prob5}, except that we \emph{do not} have a budget $B$. The goal is to select a set of facilities of minimum cost (if such a set exists) such that, after removing the edges with the probability less than $P$, all the connected components do not violate any of the two constraints defined in Problem~\ref{prob5}.
\end{problem}

\subsection{Hardness results}
\label{sec:hard}

In this section we show that problems introduced in Subsection~\ref{sec:funny_problems} are NP-hard. We  show a complete proof only for Problem~\ref{prob1}, since the  NP-hardness proofs for the other problems are  similar.

\begin{theorem}
Problem~\ref{prob1} is NP-hard.
\end{theorem}
\begin{proof}
We show a simple reduction from the Vertex Cover problem which is a classical NP-hard problem~\cite{garey1979computers}. In the (decision version of the) Vertex Cover problem the input is an undirected graph $G=(V,E)$ and an integer $k$ and the goal is to decide, if  exists, a subset $V' \subseteq V$ such that $|V'| \le k$ and for any edge $(a,b) \in E$, either $a \in V', b \in V'$ or both. Thus, given an instance of Vertex Cover that is, a graph $G = (V,E)$ and an integer $k$, we construct an instance of Problem~\ref{prob1} as follows. 
\begin{enumerate}
    \item  The input graph $G'$ of Problem~\ref{prob1} has the same vertex set $V$.
    \item The edge set is constructed as follows: for every edge $(a,b) \in E$, we set $p(a,b) = 1$, otherwise we set $p(a,b) = 0$. 
    \item The set $V^1 = \emptyset$. 

    \item The set $V^2 = \{ \{v\} \text{ } | \text{ } \forall v \in V\}$, while the cost $c_2$ of selecting any set from $V^2$ is $1$ and $r_1$ is 0 (that is, all the edges that are incident to a selected vertex are deleted).
    \item The budget $B = k$.
\end{enumerate}

Now, we show that the graph $G = (V,E)$ has a vertex cover of size at most $k$ if and only if the maximum number of connected components in the corresponding instance of Problem~\ref{prob1} is $n$. 

First, given a vertex cover $V'$, the solution of Problem~\ref{prob1} that creates $n$ connected components selects the set $\hat{V^2} = \{ \{v'\} \text{ } | \text{ } v' \in V' \}$, that is, we select the sets from $V^2$ corresponding to the vertices in $V'$. Since $V'$ is a vertex cover, any edge is incident to at least one vertex from $V'$, thus $p(a,b) = 0, \forall a,b \in V$ after selecting $\hat{V^2}$.

Conversely, given a set $\hat{V^2}$, such that $|\hat{V^2}| \le k$, we construct the set $V' = \{ v' \text{ } | \text{ } \{v'\} \in \hat{V^2} \}$. Since $n$ connected components are created after selecting $\hat{V^2}$, we know that $p(a,b) = 0, \forall a,b \in V$ (otherwise, we have a connected component with at least two vertices). Since  $p(a,b) = 0, \forall a,b \in V$, we know that for any edge $(a,b) \in E'$ that had $p(a,b) = 1$, either $\{a\} \in \hat{V^2}$ or $\{b\} \in \hat{V^2}$. Thus, $V' = \{ v' \text{ } | \text{ } \{v'\} \in \hat{V^2} \}$ is a vertex cover of $G$, completing the proof.
\end{proof}

Using a similar reduction, we can show that Problems~\ref{prob2},~\ref{prob3},~\ref{prob4},~\ref{prob5} and ~\ref{prob6} are NP-hard. Thus, we state the following corollary.

\begin{corollary}
Problems~\ref{prob2},~\ref{prob3},~\ref{prob4},~\ref{prob5} and ~\ref{prob6} are NP-hard.
\end{corollary}

%Since even the simplest variant of the problems introduced in this paper is NP-hard, during the rest of the paper we focus on approximate variants of the problem. 

%We are aware that the problem can be tackled via other methods, such as approximation algorithms (with a worst case approximation guarantee)~\cite{} or fixed parameter algorithms~\cite{} or integer linear programming solvers~\cite{}. 

\section{The framework for modeling COVID-19}
\label{sec:frame2}

The framework presented in the previous section, although promising, has the following problem. The closure of a facility might not have the same effect for all the people that are connected through that facility. Consider the following simple example: two siblings (who live in the same house) study at the same school. Then, after closing the school, in reality the two siblings still have a large probability to get in contact with each other. Thus, we introduce the following framework which captures the aforementioned example and is also simpler than the framework presented in Section~\ref{sec:frame1}.

\begin{problem}
\label{final_prob}

The input consists of a bipartite graph $G = (U \cup V, E)$. The set $U$ represents the people and the set $V$ represents the facilities. For each edge we have associated a value $p : U \times V \to [0,1]$ that represents the percentage of the time spent by a person in that facility in a day. For example, if $p(a,b) = 0.25$, then person $a$ spends $6$ hours ($0.25 \times 24\text{ hours}$) in facility $b$. Each person has associated a probability $f:U \to [0,1]$ of being infected. Each facility has associated a closure cost $c : U \to \mathbb{R}_+$. Closing a facility $v$ is equivalent with removing the edges incident to $v$. Moreover, we are given a cost $c' : U \to \mathbb{R}_+$ of isolating people. Isolating a subset of people $U'$ is equivalent with removing the edges incident to all the vertices in $U'$. Moreover, we are given a total budget $B$ for closing the facilities. 

The risk of a facility is informally the weighted (using the risk as the weight) sum of the time spent by the people in that facility. More precisely, $R:V \to \mathbb{R}_+$ is:

$$R(v) = \sum_{u \in U : (u,v) \in E} f(u) \cdot p(u,v)$$

The risk of a person $r:U \to \mathbb{R}_+$ is defined as the weighted sum spent by a person in the facilties he visits (weighted using the riks of the facility). Formally:

$$r(u) = \sum_{v \in V : (u,v) \in E} R(v) \cdot p(u,v)$$

We define $r(U)$ the vector in $\mathbb{R}^{|U|}$, that has in each component the risk of a person.

The goal is to select a set of facilities of total cost at most $B$ such that a given function $F:r(U) \to \mathbb{R}$ is minimized. In this paper we study the problem for $F$ as $\ell_1$. In other words, we aim to optimize the total risk of the people.
\end{problem}

We show that Problem~\ref{final_prob} is NP-hard even in an extremely restricted version in which there is only one person associated with a facility.

\begin{theorem}
Problem~\ref{final_prob} is NP-hard in the case  $F = \ell_1$.
\end{theorem}
\begin{proof}
The problem can be reduced to the Subset Sum problem in which the input is a set $S$ of numbers and an integer $B$ and the goal is to decide if there exists a subset of numbers from $S$ whose sum is precisely $B$. The Subset Sum problem is a famous NP-hard problem~\cite{garey1979computers}. Given an instance of the Subset Sum problem, we create an instance of Problem~\ref{final_prob} as follows. For each number $x$ in the set $S$, we create a facility $v$ of cost $c(v) = x$. Each facility $v$ has precisely one edge $(u,v) \in E$, with $f(u) = c(v)/ max_{v\in V} c(V)$ and $p(u,v) = 1$. Thus, for each pair person/facility $(u,v)$ we have $R(v) = r(u) = c(v)/ max_{v\in V}$. A set of facilities of cost $X$ gives a total risk for people of cost $$ \frac{\sum_{v \in V} c(v) - X}{\sum_{v \in V} c(v)}.$$ Therefore, if we can decide in polynomial time if the total risk incurred for the population in Problem~\ref{final_prob} is $$ \frac{\sum_{v \in V} c(v) - B}{\sum_{v \in V} c(v)},$$ then we can decide if there exist a subset of numbers from $S$ that have sum precisely $B$. Thus, Problem~\ref{final_prob} is NP-hard in the case $F = \ell_1$.
\end{proof}

\section{The algorithm}
\label{sec:algo}
In this section we provide a heuristic (approximation) algorithm for Problem~\ref{final_prob}. We test our algorithm in Section~\ref{sec:experiments} and show that it gives promising results.

Our algorithm (presented in Algorithm~\ref{alg:heuristic}) sorts the list of people and the facilities according to their efficiency (the cost of isolating/closing a person/faciltiy divided by the amount of risk the people/facilities have). Then, the algorithm aims to find the optimum division of the available budget between isolating people and closing facilities. According to our experiments (see Section~\ref{sec:experiments}) there is not an obvious correlation between the optimal value of the division of the budget (i.e., variable $Split$ in Algorithm~\ref{alg:heuristic}) and the minimum total risk. Thus, we need to iterate over all values of $Split$ in order to find a good solution. Of course, since there are infinitely many numbers between $0$ and $1$, we cannot iterate over all possible values. Choosing a larger increment  improves the running time but reduces the accuracy of the solution.

\begin{algorithm}

\begin{enumerate}
    \item Define the efficiency of a facility $v$ as $$e(v) = \frac{c(v)}{ R(v)}$$
    \item Define the efficienty of isolating a person $u$ as $$e'(u) =  \frac{c'(u)}{f(u)}$$
    \item Sort the sequence of values $e$ and $e'$ in increasing order.
    \item $MinRisk = \infty$
    \item For every value of \emph{Split} between $0$ and $100$ (in increments of $1$) do:
    \begin{enumerate}
        \item Isolate people in the order given by $e'$ until a budget of $B \cdot \frac{1}{Split}$ is reached.
        \item Close the facilities in the order given by $e$ until the budget $B$ is reached.
        \item Let $R_{Split}$ be the total risk of the population acording to this solution. If $R_{Split}< MinRisk$ then we update the value of $MinRisk$ and store the current solution.
    \end{enumerate}
    \item {\bf Output:} MinRisk and the corresponding set of people and facilities that have to be isolated/closed.
\end{enumerate}

\caption{A heuristic algorithm for Problem~\ref{final_prob}.}
\label{alg:heuristic}
\end{algorithm}

\section{Experiments}
\label{sec:experiments}

\subsection{Data generation}
\label{sec:data_gen}
In this subsection we describe how we generated our data. 

First our data generator allows two parameters as input that determine the number of facilities and the maximum size of a facility. The size of the facilities (i.e., how many people visit that facility in a day) is drawn according to a power law distribution with exponent $\alpha$ (in our experiments $\alpha$ varies between $0.8$ and $1.3$). We also select an average number of daily activities for a person (i.e., how many facilities  a person visits during one day). In our experiments the average number of activities is set between $3$ and $8$). Then, we set the number of people in a country to be the sum of all the facilities divided by the average number of facilities a person visits during one day.

In Figure~\ref{fig:1} we show an example of the distribution of the size of the facilities for $1000$ facilities each having a size between $10$ and $10000$.

\begin{figure}[h!tb]
\includegraphics[width=\textwidth,height=\textheight,keepaspectratio]{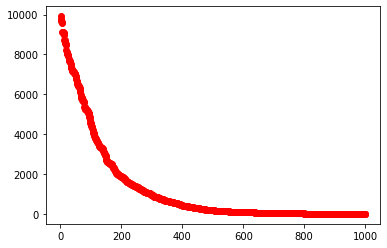}
\caption{The size of $1000$ facilities (i.e. daily number of people that visit that facility). Each facility has at least $10$ and at most $10000$ daily visitors. The number of visitors is drawn from a power law distribution with $\alpha = 1.1$.}
\label{fig:1}
\end{figure}

For each facility $v$ we select $size(v)$ people that will visit that facility uniformly at random from the population, where $size(v)$ is the size of facility $v$ that was generated previously using the power law distribution. The number of activities performed daily by each person  form a Poisson distribution (see Figure~\ref{fig:2} for an example). 

\begin{figure}[h!tb]
\includegraphics[width=\textwidth,height=\textheight,keepaspectratio]{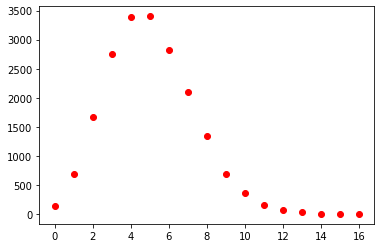}
\caption{The distribution of number of daily activities for a population of  $19\,573$ people that have on average $5$ daily activities.}
\label{fig:2}
\end{figure}

We now show how we generate the weights on the edges. For each person, we choose the time spent in each facility using an exponential distribution. 
%In Figure~\ref{fig:3} we show three examples of the time distribution among various activities. 

%\begin{figure}[h!tb]
%\includegraphics[width=\textwidth,height=\textheight,keepaspectratio]{Figure_2.png}
%\caption{The distribution of time spent performing various activities for three classes of people: a) A person that does only two activities (e.g., an old person that stays most of the time at home and goes for shopping a small amount of time. b) A moderately active person that goes to work; c) A very active person}
%\label{fig:3}
%\end{figure}
The chance that a person $i$ caries the virus, i.e., $f(i)$, is also drawn from a power law distribution with exponent $\alpha_2$. One important thing to notice is that  $\alpha_2$ influences significantly the risk of the whole population to get infected. More precisely, if $\alpha_2$ is large (that is, there are few people with high risk of carying the virus),  the risk of infection for the other people is relatively low. In Figure~\ref{fig:4} we show the risk associated to the people (calculated as shown  in Problem~\ref{final_prob}) for the values of $\alpha_2 = 4$ and $\alpha_2 = 2$. This observation, motivates us in the design of algorithm, by isolating first the persons with very high risk.

\begin{figure}[h!tb]
    \centering
    \begin{subfigure}[t]{0.5\textwidth}
        \centering
        \includegraphics[height=1.5in]{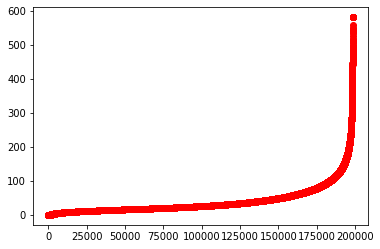}
        \caption{The case of $\alpha_2 = 2$}
    \end{subfigure}%
    ~ 
    \begin{subfigure}[t]{0.5\textwidth}
        \centering
        \includegraphics[height=1.5in]{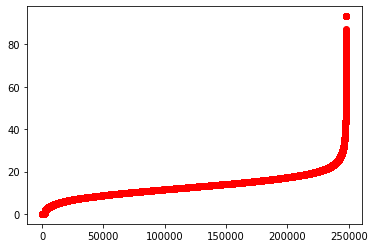}
        \caption{The case of $\alpha_2 = 4$}
    \end{subfigure}
    \caption{Comparison of the risk of the population to get infected for $\alpha_2 = 2$ and $\alpha_2 = 4$}
    \label{fig:4}
\end{figure}

Finally, we have to set the cost of isolating people and the cost of closing facilities. We choose the cost of isolating a person as a fraction of total budget available (this fraction can also be set as an input parameter in our generator). The cost of closing a facility of size $s$ is $s^x$, where $x$ is a random variable drawn according to a Gaussian distribution with mean $\mu$ and variance $\sigma$ (in our tests we vary the $\mu$ between $1.1$ and $1.2$ and $\sigma$ between $0.3$ and $0.5$). Finally, the budget is also an input parameter in the generator and we design it as a fraction of the total cost of closing the facilities, generally, between $1\%$ and $30\%$.

\subsection{Tests}

%We carry out experments for the number of facilities ranging between $200$ and $5000$ and their maximum/minimum size ranging between $15000/10$ and respectively $2000/3$. Also we vary the avereage number of daily activities between $3$ and $8$.

We carry our tests for a population of around $30\,000$ people. This population is achieved by varying the parameters in our model as: the number of facilities (between $100$ and $1\,000$), the average number of daily activities (between $3$ and $8$) and the size of each facility (between $4$ and $10\,000$). For each set of parameters we carry out $5$ tests and we choose the average risk produced by our algorithm over these $5$ tests.

The dataset size is the maximum that our hardware can handle. Nevertheless, we argue that our experiments scale to a larger population. In Figure~\ref{fig:5} we show how the risk changes if we change the number of facilities and the size of each facility: the risk has a decreasing trend as the size of our population increases, thus we believe that our algorithm is even better for larger scale instances.

\begin{figure}[h!tb]
    \centering
    \begin{subfigure}[t]{0.5\textwidth}
        \centering
        \includegraphics[height=1.5in]{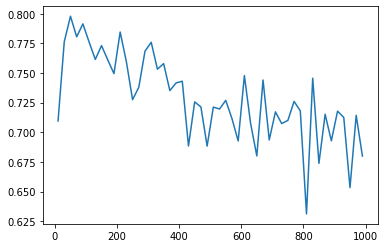}
        \caption{The horizontal axis represents the number of facilities, while the vertical axis represents the risk improvement. The size of each facility is between $4$ and $1\,000$.}
    \end{subfigure}%
    ~ 
    \begin{subfigure}[t]{0.5\textwidth}
        \centering
        \includegraphics[height=1.5in]{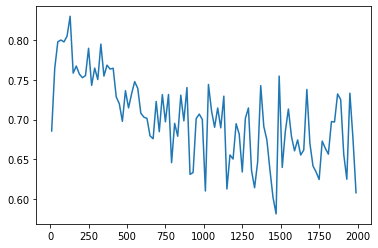}
        \caption{The horizontal axis represents the size of each facility, while the vertical axis represents  the risk improvement. The number of facilities is $500$ and the average number of daily activities performed by a person is $4$.}
    \end{subfigure}
    \caption{The improvement in the population risk (i.e., the risk of the population after our algorithm, divided by the risk before the run of our algorithm) compared with the size of the instance. The average number of daily activities performed by a person is $4$. The budget allocated is $5\%$ of the ammount necessary to close all facilities. With this budget we are able to quarantine $5\%$ of the population. Then, we have $\alpha = 1.1$ and $\alpha_2 = 2$.
    Observe that the improvement is bigger as the population increases. }
    \label{fig:5}
\end{figure}

Next, we show how the split of the budget between isolating people and closing facilities influences the total risk. In our tests we have $500$ facilties between $4$ and $1\,000$ people, each person performs on average $4$ activities per day and we have $\alpha = 1.1$, $\alpha_2 = 2$. The cost of the exponent of the random variable that determines the cost of closing the facilities is drawn from a normal distribution with $\mu = 1.1$ and $\sigma = 0.4$ (Figure~\ref{fig:6}a and Figure~\ref{fig:7}a), $\sigma = 0.5$ (Figure~\ref{fig:6}b and Figure~\ref{fig:7}b). The budget is $10\%$ of the cost of closing all facilities in Figure~\ref{fig:6} and $1\%$ in Figure~\ref{fig:7}. This budget suffices to isolate $10\%$ of the population, respectively $1\%$.  

 Notice that with  a budget of only $1\%$ from the cost of closing all facilities, we are able to lower the risk to less than $20\%$ of the original risk (Figure~\ref{fig:7}b).

%Notice that even with such a small budget the risk after running the algorithm is significantly reduced (at around $10\%$ of the original risk).

\begin{figure}[h!tb]
    \centering
    \begin{subfigure}[t]{0.5\textwidth}
        \centering
        \includegraphics[height=1.5in]{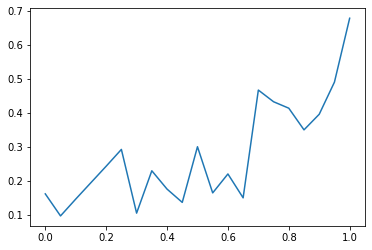}
        \caption{A $10\%$ of the total cost of closing the facilities and $\sigma = 0.4$}
    \end{subfigure}%
    ~ 
    \begin{subfigure}[t]{0.5\textwidth}
        \centering
        \includegraphics[height=1.5in]{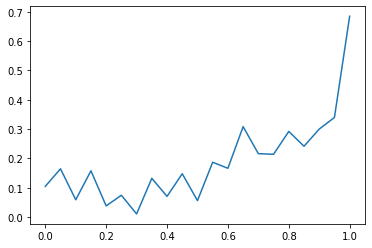}
        \caption{A $10\%$ of the total cost of closing the facilities and $\sigma = 0.5$}
    \end{subfigure}
    \caption{On the $x$ axis is the percentage of the total budget allocated to isolating people. On the $y$ axis there is the ratio of the risks  before/after running the algorithm.}
    \label{fig:6}
\end{figure}

\begin{figure}[h!tb]
    \centering
    \begin{subfigure}[t]{0.5\textwidth}
        \centering
        \includegraphics[height=1.5in]{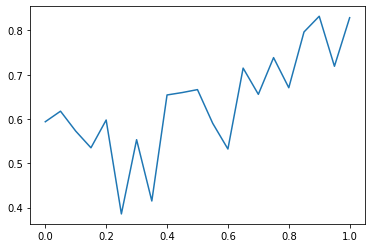}
        \caption{A $1\%$ of the total cost of closing the facilities and $\sigma = 0.4$}
    \end{subfigure}%
    ~
    \begin{subfigure}[t]{0.5\textwidth}
        \centering
        \includegraphics[height=1.5in]{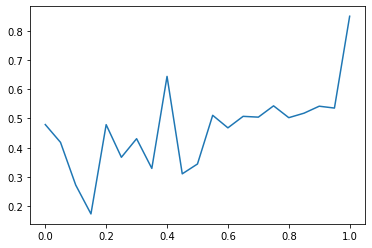}
        \caption{A $1\%$ of the total cost of closing the facilities and $\sigma = 0.5$}
    \end{subfigure}%
    \caption{On the $x$ axis is the percentage of the total budget allocated to isolating people. On the $y$ axis there is the ratio of the risks  before/after running the algorithm.}
    \label{fig:7}
\end{figure}

Finally, we tested how does the risk decrease if we take actions quickly. More precisely, we vary $\alpha_2$ which is the power law exponent that determines the percentage of people that are likely to be already infected. However, we did not notice any major influence of this factor in the total risk if the infection proportion is drawn according to a power law distribution. 

Our algorithm was implemented in Python and the tests were carried out on a $2013$ MacBook Pro with $2.4$ GHz Quad-Core Intel Core $i7$, and $8$GB RAM. The code used for testing and generating data is available on GitHub~\cite{Popa2020}.

\section{Integer programming formulation (ILP)}
\label{sec:ilp}
We are aware that NP-hard problems can be formulated as integer linear programms (ILPs) and then solved using dedicated software (e.g., Gurobi~\cite{gurobi2018gurobi} or CPLEX~\cite{cplex2009v12}). However, the ILP approach is feasible only for small datasets, which is not the case in this paper.

Nevertheless, for the sake of completeness, in this section we present an ILP formulation of Problem~\ref{final_prob}. We have the following types of variables. First we have two variables $x_u$, for each person $u \in U$ and $y_v$, for each facility $v \in V$. Both variables are either $0$ or $1$ depending whether the corresponding person is isolated or not,  respectively whether the corresponding facility is closed or not. Then, we have variables $r_u$ and, respectively, $R_v$, that determine the risk of a person, respectively of a facility. The ILP is presented below.

\begin{equation*}
\begin{array}{ll@{}ll}
\text{minimize}  &  & \\
& \displaystyle\sum\limits_{u \in U} r(u) x_u    & \\

\text{subject to}& & \\

& \displaystyle\sum\limits_{u \in U} c'(u) x_u + \displaystyle\sum\limits_{v \in V} c(v) y_v \leq B   & \\
& \displaystyle\sum\limits_{(u,v) \in E} f(u) p(u,v) x_u  =  R_v  & \forall v \in V \\ 
& \displaystyle\sum\limits_{(u,v) \in E} p(u,v)R_v y_v   =  r_u  & \forall u \in U \\ 
& x_{u}, x_{v} \in \{0,1\}  & \forall u \in U, v \in V
 
\end{array}
\end{equation*}

\section{Conclusions and future work}
\label{sec:conclusions}
In this paper we presented a model and an algorithm that aims to help authorities to take more efficient decisions in the fight with COVID-19. Naturally, the most stringent open problem is to test and validate the model and the algorithm on real data. People have a huge mobility nowadays and it is impossible to create a model which is fully accurate. Nevertheless, based on our tests we believe that our model is capable of capturing the most important features of the current situation. 

Also, a natural open problem is to tune the input parameters: the probabilities $p$ in the input graph, the cost of closing facilities and isolating people $c$ and $c'$ and the contagion risk associated with each person. 

Since the economy is under severe pressure under the current lockdown, we expect that some ease of the restrictions will happen soon. Thus, we are hopeful that our model will give the authorities some insight in taking the best decisions. Moreover, as we can see from our experiments, even with a very small budget (sometimes as low as $1\%$ of the total cost necessary to lock down the entire economy), the risk of infection can be decreased significantly. Thus, we strongly believe that, with wise decisions, it is possible to stop the spread of COVID-19 without an economic collapse.

\paragraph*{Acknowledgements.} I would like to thank Ramona Georgescu, P\' eter Bir\' o,  Radu Mincu and Lucian-Ionut Gavril\u a for extremely useful discussions.

\bibliographystyle{plain}
\bibliography{bib}

\end{document}